\documentclass[conference,a4paper]{IEEEtran}

\usepackage{amssymb}
\usepackage{amsthm}

\usepackage[utf8]{inputenc}
\usepackage[T1]{fontenc}
\usepackage{url}
\usepackage{ifthen}
\usepackage{cite}
\usepackage{amsfonts,amsthm,bm}
\usepackage[cmex10]{amsmath} 
\usepackage{comment}

\usepackage{graphicx}
\newtheorem{theorem}{Theorem}

\newtheorem{corollary}{Corollary}

	\usepackage{cite}
	\usepackage[ruled,vlined,linesnumbered]{algorithm2e}

\newcommand{\defeq}{\triangleq}
\newcommand{\Expect}{{\rm I\kern-.3em E}}

\newtheorem{example}{{\em Example}}

\newcommand\blfootnote[1]{%
  \begingroup
  \renewcommand\thefootnote{}\footnote{#1}%
  \addtocounter{footnote}{-1}%
  \endgroup
}

\author{Eleftherios Lampiris, Jingjing Zhang, Osvaldo Simeone, Petros Elia}

\usepackage{color}

\title{Does Coded Caching Suffer from Uneven Channels?}
\title{NOMA with Caches}
\title{Fundamental Limits of Wireless Caching \\ under Uneven-Capacity Channels}
\date{}

\begin{document}

\maketitle

\begin{abstract}\blfootnote{Eleftherios is with the Electrical Engineering and Computer Science Department, Technische Universit\"at Berlin, 10587 Berlin, Germany (lampiris@tu-berlin.de). Petros is with the Communication Systems Department at EURECOM, Sophia Antipolis, 06410, France (elia@eurecom.fr). Their work is supported by the ANR project ECOLOGICAL-BITS-AND-FLOPS and the European Research Council under the EU Horizon 2020 research and innovation program / ERC grant agreement no. 725929. (ERC project DUALITY). Jingjing and Osvaldo are with the Department of Informatics, King's College London, London, UK (jingjing.1.zhang@kcl.ac.uk, osvaldo.simeone@kcl.ac.uk). Their work is supported by the European Research Council (ERC) under the European Union’s Horizon 2020 Research and Innovation Programme (Grant Agreement No. 725731). This work was conducted while Eleftherios was employed by EURECOM.}
    This work identifies the fundamental limits of cache-aided coded multicasting in the presence of the well-known `worst-user' bottleneck. This stems from the presence of receiving users with uneven channel capacities, which often forces the rate of transmission of each multicasting message to be reduced to that of the slowest user. This bottleneck, which can be detrimental in general wireless broadcast settings, motivates the analysis of coded caching over a standard Single-Input-Single-Output (SISO) Broadcast Channel (BC) with $K$ cache-aided receivers, each with a generally different channel capacity. For this setting, we design a communication algorithm that is based on superposition coding that capitalizes on the realization that the user with the worst channel may not be the real bottleneck of communication. We then proceed to provide a converse that shows the algorithm to be near optimal, identifying the fundamental limits of this setting within a multiplicative factor of $4$. Interestingly, the result reveals that, even if several users are experiencing channels with reduced capacity, the system can achieve the same optimal delivery time that would be achievable if all users enjoyed maximal capacity.
\end{abstract}

\section{Introduction}

The seminal work in \cite{maddah2014fundamental} showed how adding caches to receiving nodes can substantially reduce the time required to deliver content. Specifically, reference \cite{maddah2014fundamental} studied the case in which a transmitter with access to a library of $N$ unit-sized files serves -- via a wired, single-stream, unit-capacity bottleneck link -- $K$ cache-aided receivers/users. Each user is equipped with a cache of size equal to a fraction $\gamma\in[0,1]$ of the size of the library, so that $K\gamma$ is the cumulative cache size normalized by the library size. For this setting, the authors of \cite{maddah2014fundamental} proposed a novel cache placement algorithm and a novel multicast transmission policy that delivers any set of $K$ files to the receivers with (normalized) delay at most
\begin{align}\label{eqMNperformance}
    T_{MN}=\frac{K(1-\gamma)}{K\gamma+1}
\end{align}
thus revealing a speed-up factor of $K\gamma+1$ compared to the delay $K(1-\gamma)$ corresponding to a standard scheme that serves each user in turn.

The delay \eqref{eqMNperformance} is obtained by a \emph{coded caching} approach that is based on the transmission of a sequence of multicast messages that convey information to several users at a time (even if these users requested different content), with users decoding their desired information by means of cache-aided interference cancellation. In this scheme, each multicast message consists of a XOR $X_{\sigma}$ that carries information to a subset $\sigma \subset [K]\defeq [1,2,\dots,K]$ of $|\sigma|=K\gamma+1$ users at a time.

While the promised speedup factor of $K\gamma+1$ in \eqref{eqMNperformance} is proportional to the normalized \textit{cumulative cache size of the network}, it was quickly realized that a variety of bottlenecks severely hamper this performance. These include the subpacketization bottleneck~\cite{shanmugam2016finite,lampiris2018addingJSAC,yan2017placement,tang2018coded,shangguan2018centralized,KrishnaProjective2019,MingyueSubpacketization2019}, the uneven cache sizes bottleneck \cite{ibrahim2019coded,asadi2018centralized,AmiriDecentralizedUnevenCachesTransComm,lampiris2018full}, and the bottleneck studied here that arises from uneven channel capacities between the transmitter and the users. This last bottleneck is particularly relevant in wireless scenarios with multicasting. Such networks produce ``slower'' users that can force the multicast rates to be reduced down to a level that can be decoded by these users. This can diminish the coded caching gains and could pose a serious limitation to any effort to implement cache-aided coded multicasting in wireless settings. 

\begin{example}
    Let us consider the wireless Single-Input-Single-Output (SISO) Broadcast Channel (BC) with $K$ users, each equipped with a cache of normalized size $\gamma$, and let us further assume that all users have maximal normalized unit capacity, except for one user that has a normalized link capacity equal to $\frac{1}{K}+\gamma<1$. It is easy to see that a (naive) transmission of the sequence of the XORs from \cite{maddah2014fundamental} would induce the delay 
    \begin{align}\label{doubledelay}
        T
        &=\frac{1-\gamma}{\frac{1+K\gamma}{K}}+\frac{(K-K\gamma-1)(1-\gamma)}{1+K\gamma}\\
        &=2T_{MN}-(1-\gamma)\approx 2 T_{MN}
    \end{align}
    which is approximately double the delay $T_{MN}$ in \eqref{eqMNperformance} that we would have if all users enjoyed unit normalized link capacities.  It is also worth noting that approximately the same delay $T$ in \eqref{doubledelay} would be obtained if we treated the slow user separately from the rest using time sharing. Essentially, whether with a naive or with a separated approach that excludes the slow user from coded caching, a single slow user can cause the worst-case delivery time to double, and the overall multicasting gain to be cut in half.
 \end{example}

\subsection{Related Work} The importance of the uneven-channel bottleneck in coded caching has been acknowledged in a large number of recent works that seek to understand and ameliorate this limitation \cite{zhengWirelessVideoGlobeCom2016,zhangTopologicalISIT2017,ngoScalableTransWireless2018,destounisAlphaFairWiOpt2017,7558129,8036265,8371012,8006856,8359316,lampirisNoCsitISIT,piovanoGDoFMISOBCTransIT2019,ShariatpanahiPhysicalLayer2019TransIT,tolli2017multi,tolliMulticast2018ISIT,zhao2019low,salehi2019subpacketization}. For example, reference \cite{zhengWirelessVideoGlobeCom2016} focuses on the uneven link-capacity SISO BC where each user experiences a distinct channel strength, and proposes algorithms that outperform the naive implementation of the algorithm of \cite{maddah2014fundamental} whereby each coded message is transmitted at a rate equal to the rate of the worst user whose message appears in the corresponding XOR operation. 
Under a similar setting, the work in \cite{destounisAlphaFairWiOpt2017} considered feedback-aided user selection that can maximize the sum-rate as well as increase a fairness criterion that ensures that each user receives their requested file in a timely manner. In the related context of the erasure BC where users have uneven probabilities of erasures, references \cite{7558129} and \cite{8036265} showed how an erasure at some users can be exploited as side information at the remaining users in order to increase system performance. Related work can also be found in~\cite{8371012,8006856,8359316}.

The uneven-capacity bottleneck was also studied in the presence of multiple transmit antennas \cite{ngoScalableTransWireless2018,8374074}. Reference \cite{ngoScalableTransWireless2018} exploited transmit diversity to ameliorate the impact of the worst-user capacity, and showed that employing $\mathcal{O}(\ln K)$ transmit antennas can allow for a transmission sum-rate that scales with $K$. Similarly, the work in \cite{8374074} considered multiple transmit and multiple receive antennas, and designed topology-dependent cache-placement to ameliorate the worst-user effect.

In a related line of work, the papers \cite{lampirisNoCsitISIT} and \cite{piovanoGDoFMISOBCTransIT2019} studied the cache-aided topological interference channel where $K$ cache-aided transmitters are connected to $K$ cache-aided receivers, and each transmitter is connected to one receiver via a direct ``strong'' link and to each of the other receivers via ``weak'' links. Under the assumption of no channel state information at the transmitters (CSIT), the authors showed how the lack of CSIT can be ameliorated by exploiting the topology of the channel and the multicast nature of the transmissions.

Recently, significant effort has been made toward understanding the behavior of coded caching in the finite Signal-to-Noise Ratio (SNR) regime with realistic (and thus often uneven) channel qualities. In this direction, the work in \cite{ShariatpanahiPhysicalLayer2019TransIT} showed that a single-stream coded caching message beamformed by an appropriate transmit vector can outperform some existing multi-stream coded caching methods in the low-SNR regime, while references \cite{tolli2017multi,tolliMulticast2018ISIT} (see also \cite{zhao2019low}) revealed the importance of jointly considering caching with multicast beamformer design. Moreover, the work in \cite{salehi2019subpacketization} studied the connection between rate and subpacketization in the multi-antenna environment, accounting for the unevenness naturally brought about by fading.

Our work is in the spirit of all the above papers, and it can be seen specifically as an extension of \cite{zhangTopologicalISIT2017}. This reference considered a specific binary topological case, for which it proposed a two-level superposition-based transmission scheme to alleviate the worst-user bottleneck. Further, a similar approach has been proposed in the work of \cite{amiriEnergyMinimizationJSAC2018}, where though the closely related scheme places focus on minimizing the power.

\subsection{Overview of Results} 
In this paper, we study a cache-aided SISO BC where each receiver $k$ experiences a link of some normalized capacity $\alpha_k\in[0,1]$. We establish the optimal worst-case delivery time $T(K,\gamma,\{\alpha_k\})$ within a factor of at most $4$ for any number of $K$ users, fractional cache capacity $\gamma$, and capacity set $\{\alpha_k\}$. Key to this result is a new algorithm that uses superposition coding, where (assuming without loss of generality that the users are labeled from weaker to stronger, i.e., such that $\alpha_{k}\le \alpha_{k+1}$) we split the power into $K-K\gamma-1$ layers, and in layer $k$, we transmit \textit{only} XORs whose weakest user is user $k$. While this design indeed encodes some XORs at lower rates (matching the capacity of the worst user for that message), it also allows the simultaneous transmission of other XORs in the remaining power layers.
The main result reveals that the optimal performance \eqref{eqMNperformance} achievable when $\alpha_k=1$, for all $k\in[K]\defeq [1,2,\dots,K]$, is in fact achievable even if each user $k$ has reduced link capacity such that the condition
\begin{align}
	\alpha_{k} \gtrsim 1-e^{-k\gamma},~~ \forall k\in[K]
\end{align}
is satisfied. This quantifies the intuitive fact that systems with smaller caches can be better immune to the negative effects of channel unevenness.

\section{System Model} 

We consider the $K$-user wireless SISO BC, with the transmitter having access to a library of $N$ files $\{W^n\}_{n=1}^{N}$, each of normalized unit size, and the $K$ receivers having a cache whose size is equal to a fraction $\gamma \in[0,1]$ of the library size. Communication takes place in two distinct phases, namely the pre-fetching and the delivery phases. In the first phase, the caches of the users are filled with content from the library without any knowledge of future requests or of channel capacities.  Then, during the delivery phase, each user $k$ requests\footnote{We are interested in the worse-case delivery time and thus we will assume that each user will ask for a different file.} a single file $W^{d_k}$, after which the transmitter -- with knowledge of the requests and the link capacities -- delivers the requested content. 

After transmission, at each user $k\in [K]$, the received signal takes the form 
\begin{align}
    y_{k}= \sqrt{P^{\alpha_k}}h_{k} x +z_{k},
\end{align}
where $P$ represents the transmitting power; $x\in\mathbb{C}$ is the power-normalized transmitted signal satisfying $\mathbb{E}\{|x|^{2}\}\le 1$; $h_{k}\in\mathbb{C}$ is the channel coefficient of user $k$; $z_{k}\thicksim \mathbb{C}\mathcal{N}(0,1)$ represents the Gaussian noise at user $k$; and $ \alpha_k\in(0,1]$ is such that at each user $k\in[K]$ the average SNR equals
\begin{align} \label{con}
   \mathbb{E}\{|y_{k}|^{2}\}=P^{\alpha_k}.
\end{align}
Under the simplified Generalized Degrees of Freedom (GDoF) framework of \cite{jafarGDoFTransIT2010,davoodiAlignedImageSetsTransIT2016,gholamiGeneralized2017TransIT}, condition \eqref{con} amounts to a (normalized, by a factor $\log P$) user rate of $r_{k}=\alpha_k  \in[0,1]$. Without loss of generality, $\alpha=1$ corresponds to the highest possible channel strength.  
We assume an arbitrary set of such normalized capacities $\boldsymbol{\alpha}\triangleq\{\alpha_k\}_{k=1}^{K}$ and we assume them without loss of generality to be ordered in ascending order ($\alpha_{k}\le \alpha_{k+1}$).

The objective is to design the caching and communication scheme that minimizes the worst-case delivery time $T(K,\gamma,\boldsymbol{\alpha})$ for any capacity vector $\boldsymbol{\alpha}$.

\section{Main Results}
Before presenting the main results, we remind the reader that the naive implementation of coded caching which sequentially transmits the sequence of XORs $X_{\sigma}$ to all subsets $\sigma\in[K]$ of $\ |\sigma| = K\gamma+1$ users, requires a worst-case delivery time 
\begin{align}\label{eqNaiveCodedCaching}
    T_{uc}(K,\gamma,\boldsymbol{\alpha})=\frac{1}{\binom{K}{K\gamma}} \sum_{\sigma\subseteq[K],~|\sigma|=K\gamma+1} \max_{i\in\sigma}\left\{\frac{1}{\alpha_{i}}\right\}.
\end{align}
This follows since this conventional uncoded scheme allocates, for each XOR $X_{\sigma}$, a transmission time $
        T_{\sigma}=\max_{w\in\sigma}\left\{\frac{1}{\alpha_{w}}\right\}
    $ to allow the weakest user in $\sigma$ to decode the message\footnote{This is a well known expression that has been calculated in a variety of works such as in  \cite{zhengWirelessVideoGlobeCom2016,ShariatpanahiPhysicalLayer2019TransIT}.}.

\nocite{lampirisSubpacketizationCsitSPAWC,lampiris2018lowCSIT}

We now proceed with the main result.

\begin{theorem} \label{the:ach}
    In the $K$-user SISO BC with receiver channel strengths $\big\{\alpha_{k}\big\}_{k=1}^{K}(\alpha_{k}\le\alpha_{k+1})$ and with receivers equipped with a cache of normalized size $\gamma$, the worst-case delivery time
    \begin{align}\label{eqLosslessCompletionTime}
        T_{sc}(K,\gamma,\boldsymbol{\alpha})= \max_{w\in[K]} \left\{   \frac{1}{\alpha_{w}}\cdot   \frac{\binom{K}{K\gamma+1}-\binom{K-w}{K\gamma+1}}{\binom{K}{K\gamma}}\right\}
    \end{align}
    is achievable and is within a multiplicative factor of at most $4$ from the optimal delay $T^*(K,\gamma,\boldsymbol\alpha)$.
\end{theorem}
\begin{proof}
    The achievability part of the scheme is described as  Algorithm~\ref{algLosslessDelivery} in Section~\ref{secLosslessTransmission}, while the converse and the derivation of the gap to optimal are presented in Section~\ref{secGapGeneral}.
\end{proof}

One of the main conclusions from the above result is summarized in the following corollary. 
\begin{corollary}\label{corDegradedButOptimal}
         In the same $K$-user SISO BC with $\gamma$-sized caches and (ordered) capacities $\big\{\alpha_{k}\big\}_{k=1}^{K}$, the baseline performance 
         \begin{align}
         	T(K,\gamma,\boldsymbol{\alpha}=\mathbf{1})=T_{MN}=\frac{K(1-\gamma)}{1+K\gamma}
         \end{align} associated to the ideal case $\alpha_{k}=1 \ \forall k\in [K]$, can be achieved even if the capacities satisfy the inequalities  
     \begin{align}\label{eqChannelUnevennessStrengths}
         \alpha_{k}\ge \alpha_{th,k} \defeq 1-\frac{\binom{K-k}{K\gamma+1}}{\binom{K}{K\gamma+1}}\approx 1-e^{-k\gamma},~~ \text{for all}~ k\in[K]. 
     \end{align}
\end{corollary}

\begin{proof}
The proof is direct from Eq.~\eqref{eqLosslessCompletionTime}, after using the Sterling approximation $\binom{n}{k}\approx\left(\frac{n}{k}\right)^{k}$ and the limit 
\begin{align}
	\lim_{K\to\infty}\left(1-\frac{b}{K}\right)^{K}=e^{-b}.
\end{align}
\end{proof}
Given any user $k$, $\alpha_{th,k}=1-\frac{\binom{K-k}{K\gamma+1}}{\binom{K}{K\gamma+1}} $ provides a threshold channel capacity that allows the algorithm to achieve the baseline unit-capacity performance $T_{MN}$.

\begin{figure}[t!]
\centering
\includegraphics[width=0.85\columnwidth]{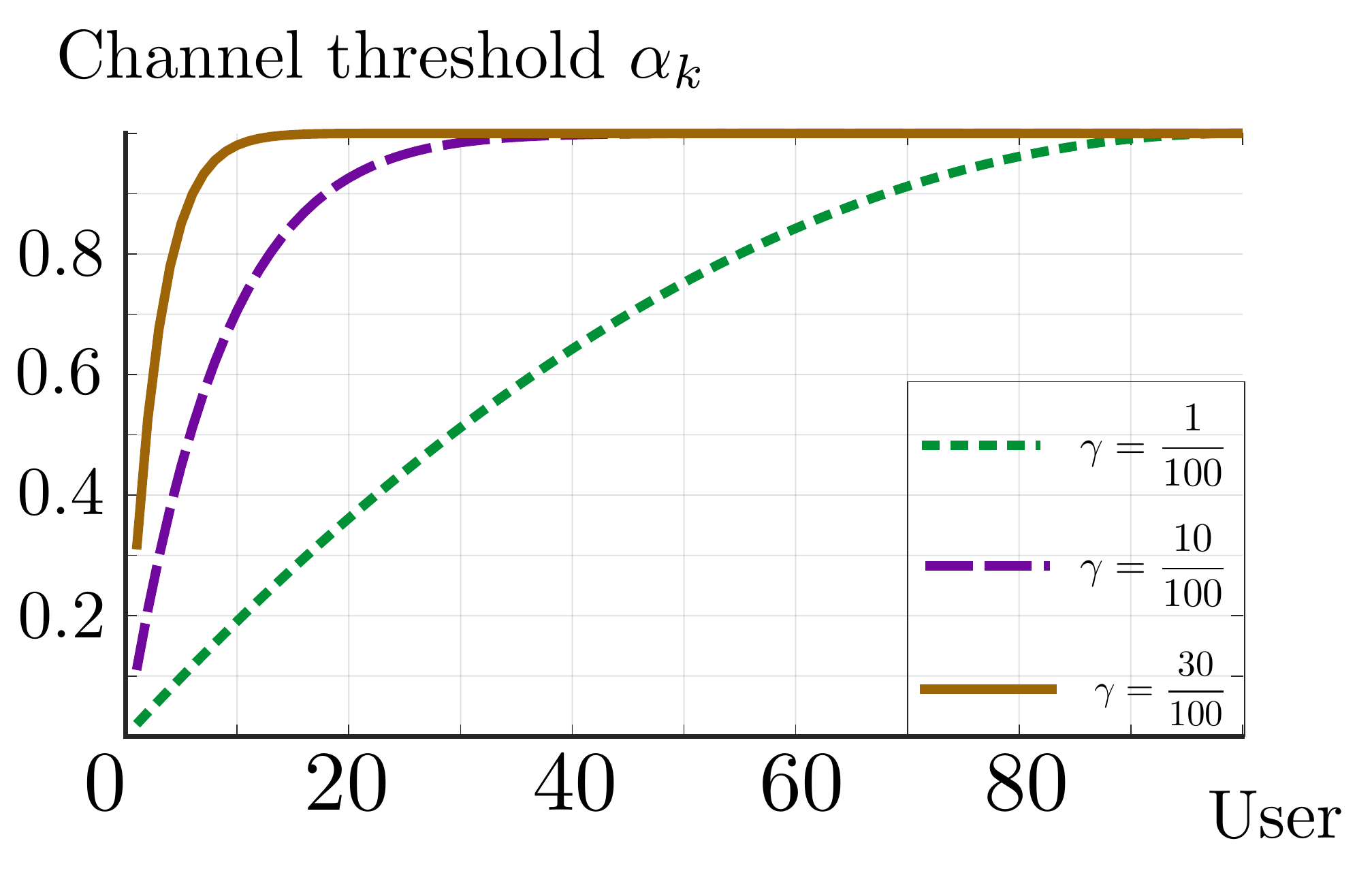}
\caption{The plot presents the threshold $\alpha_{th,k}$ for the case of $K=100$ users. We can see that as $\gamma$ decreases, an ever increasing fraction of users can have a further reduced channel capacity without any performance degradation with respect to the maximal-capacity delay.}
\label{figThreshold}
\end{figure}

\section{Placement and Delivery Algorithms}\label{secLosslessTransmission}

We here present the superposition-based communication scheme with the corresponding cache placement, transmission, and decoding that achieves the delay in Theorem~\ref{the:ach}.

\subsection{Cache Placement}

During the placement phase, we apply directly the placement algorithm of \cite{maddah2014fundamental} without exploiting any knowledge of the channel capacities. To this end, each file $W^{n},n\in[N]$, is subpacketized into $S=\binom{K}{K\gamma}$ subfiles \begin{align}
    W^{n}\to \{ W^{n}_{\tau}, \tau\subset[K],~|\tau|=K\gamma    \}
\end{align}
and the cache $\mathcal{Z}_{k}$ of user $k$ is filled as
\begin{align} \label{cache}
    \mathcal{Z}_{k}=\left\{ W^{n}_{\tau} : \tau\subset [K],\forall n\in[N]\right\}
\end{align}
which, as can easily be shown, adheres to the cache-size constraint.

\subsection{Delivery Algorithm}

After each user $k\in[K]$ requests a file $W^{d_k}$ as in \cite{maddah2014fundamental}, the transmitter delivers the $\binom{K}{K\gamma+1}$ XORs
\begin{align}
 X_{\sigma}=\bigoplus_{k\in\sigma} W^{d_{k}}_{\sigma\setminus\{k\}}
\end{align}
for all subsets $\sigma$ of users of size $|\sigma|=K\gamma+1$. To this end, in every communication slot, we split the available transmission power into $K-K\gamma-1$ ``power layers''. In power layer $k$ we encode XORs from the set 
\begin{align}
\mathcal{X}_{k}\defeq \{X_{\sigma}:\min\{\sigma\}=k\}.
\end{align}
This contains all the XORs intended for set of users $\sigma$ for which the slowest user is user $k$ i.e., all the XORs intended for user $k$ except those desired by any user whose channel is weaker than user $k$. 
It can be easily shown\footnote{The last equality follows directly from Pascal's triangle.} that the sets $\mathcal{X}_{k}$ are disjoint; that for any $k\le K-K\gamma-1$, we have 
\begin{align}
	|\mathcal{X}_{k}|=\binom{K-k+1}{K\gamma+1}-\binom{K-k}{K\gamma+1}=\binom{K-k}{K\gamma}
\end{align}
XOR messages in power layer $k$ and that the total number of XOR messages in the first $k$ power layers is 
\begin{align}
	\left| \bigcup_{m=1}^{k}  \mathcal{X}_{m} \right|=\binom{K}{K\gamma+1}-\binom{K-k}{K\gamma+1}.
\end{align}

For example, Layer $1$ (which will correspond to the highest-powered layer) contains all the XORs in set $\mathcal{X}_{1}$ i.e., all the XORs that are intended for the weakest user (user $1$). Similarly Layer $2$ will contain the XORs from  $\mathcal{X}_{2}$, i.e., those XORs that are intended for user $2$, but not for user 1, and so on. 
The power allocation for each XOR is designed so that the weakest user of the XOR can decode it, implying that any other user that needs to decode that same XOR is able to do so. The chosen power allocation seeks to minimize the overall delay.

\begin{algorithm}[th!]\caption{Delivery based on Superposition Coding}\label{algLosslessDelivery}
{
Let $\alpha_{k}\le \alpha_{k+1},~\forall k\in[K]$\\
Find $w\in[K]$ such that
\begin{align}\label{eqMaximizationFactor}
	    w\!=\!\arg\max_{k\in[K]}\left\{ \frac{\binom{K}{K\gamma+1}\!-\!\binom{K-k}{K\gamma+1}}{\alpha_{k}}\right\}.
\end{align}\\
Set $\beta_{0}=0$ and for $k\in[K-K\gamma-1]$ set
\begin{align}
    \beta_{k}=\frac{\left| \cup_{i=1}^{k} \mathcal{X}_{k}\right|}{\left| \cup_{i=1}^{w} \mathcal{X}_{k}\right|}\alpha_{w}=\frac{\binom{K}{K\gamma+1}-\binom{K-k}{K\gamma+1}}{\binom{K}{K\gamma+1}-\binom{K-w}{K\gamma+1}}\alpha_{w}.
\end{align}
\For{all $k\in[K-K\gamma-1]$}{
	Encode $x_{k}$ selected from $\mathcal{X}_{k}$ without replacement\\
	with power
\begin{align}\label{eqPowerAllocation}
    P_{k}=P^{-\beta_{k-1}}- P^{-\beta_{k}}
\end{align}\\
	and rate
\begin{align}\label{eqRateAllocation}
	r_{k}= \beta_{k}-\beta_{k-1}=\frac{\binom{K-k}{K\gamma}}{\binom{K}{K\gamma+1}-\binom{K-w}{K\gamma+1}}\alpha_{w}	.
\end{align}
}
Transmit $x_{k}, \forall k\in[K]$ simultaneously.
}
\end{algorithm}

The process is described in the form of pseudo-code in Algorithm~\ref{algLosslessDelivery}. The algorithm begins by identifying (Step $2$) the bottleneck user
\begin{align}
	    w\!=\!\arg\max_{k\in[K]}\left\{ \frac{\binom{K}{K\gamma+1}\!-\!\binom{K-k}{K\gamma+1}}{\alpha_{k}}\right\} .
\end{align}
This is defined as the user $k$ that takes the longest time to decode all power layers from 1 to $k$.
Then Step~$3$ calculates the power layer coefficients $\beta_{i}, i\in \{0,1,.., K-K\gamma-1\}$ for each power layer as explained below. In Step~$4$, for every $k\in[K-K\gamma-1]$, a new XOR is selected from set $\mathcal{X}_{k}$, and is encoded in message $x_{k}$, with power $P_{k}=P^{-\beta_{k-1}}- P^{-\beta_{k}}$ (Step~$5$) and rate $\frac{\binom{K-k}{K\gamma}}{\binom{K}{K\gamma+1}-\binom{K-w}{K\gamma+1}}\alpha_{w}$ (Step~$6$).
 Finally in Step~$7$ all the $x_{k}, \forall k\in[K]$ are transmitted simultaneously using superposition coding.

\subsection{Decoding}

In the received signal 
\begin{align}
	y_{k} =& h_{k} \sqrt{P^{\alpha_{k}}} \sum_{m_{1}=1}^{k} x_{m_{1}} +  h_{k} \sqrt{P^{\alpha_{k}}} \sum_{m_{2}=k+1}^{K-K\gamma-1} x_{m_{2}}
\end{align}
at user $k \in[ K]$, the second term $\sum_{m_{2}=k+1}^{K-K\gamma-1} x_{m_{2}}$ contains the lower power layers, which carry no valuable information for user $k$ and are treated as noise. This part of the message is transmitted with power $P^{-\beta_{k}}$, where $\beta_{k}=\frac{\binom{K}{K\gamma+1}-\binom{K-k}{K\gamma+1}}{\binom{K}{K\gamma+1}-\binom{K-w}{K\gamma+1}}\alpha_{w}$. Due to the power and rate allocation for each of these messages (cf.~Eq.~\eqref{eqPowerAllocation} and Eq.~\eqref{eqRateAllocation}),
using successive interference cancellation\footnote{In successive interference cancellation, a user first decodes the highest powered message by treating the remaining messages as noise, then proceeds to remove this -- known at this point -- message and decodes the second message by treating the remaining as noise, and so on until all messages have been decoded.} (SIC), receiver $k$ can decode the first term that encodes the messages that potentially contain information that is valuable for user $k$.

\subsection{Delay Calculation}

The total delay of the scheme is  
\begin{align}
	T_{sc}(K,\gamma,\boldsymbol{\alpha}) &= \max_{k\in[K-K\gamma-1]} \left\{ \frac{ | \mathcal{X}_{k} | }{\binom{K}{K\gamma}}\cdot \frac{1}{r_{k}} \right\}\\
	&=    \frac{1}{\alpha_{w}}\cdot   \frac{\binom{K}{K\gamma+1}-\binom{K-w}{K\gamma+1}}{\binom{K}{K\gamma}}.
\end{align}
This corresponds to the maximum delay required to deliver all XORs $X_{\sigma}\in\mathcal{X}_{k}$ across all values of $k\in[K-K\gamma-1]$.

\section{Converse and Gap to Optimality}\label{secGapGeneral}

In this section, we provide a lower bound on the optimal delay for any given set of parameters
$K,\gamma,\boldsymbol{\alpha}$, and then we prove that the achievable delay $T_{sc}\triangleq T_{sc}(K,\gamma,\boldsymbol\alpha)$ from Theorem \ref{the:ach} is within a factor of at most $4$ from the optimal delay $T^*(k,\gamma,\boldsymbol\alpha)$. 

To lower bound the minimum delay $T^*(k,\gamma,\boldsymbol\alpha)$, we consider an augmented system where the capacities of the first $w$ users, with $w$ selected as \eqref{eqMaximizationFactor}, are increased to $\alpha_{k}=\alpha_{w}\triangleq \alpha$, for all $ k\in[w]$, while the capacities of the remaining users are increased to $1$. For such a system, the delay is lower bounded as
    \begin{align}
        T_{\text{aug}}\ge\overbrace{\frac{1}{\alpha}}^{t_{1}}\overbrace{\frac{1}{2}\frac{w(1-\gamma)}{1+w \gamma}}^{t_{2}},
    \end{align}
    where term $t_{1}$ corresponds to the channel capacity of the first $w$ users, while term $t_{2}$ corresponds to a lower bound on the minimum possible worst-case delivery time\footnote{In fact, as we know from \cite{yuFactorOf2TransIT2019}, this factor is slightly smaller than $\frac{1}{2}$.} associated to a system with $w$ cache-aided users (cf.~\cite{yuFactorOf2TransIT2019}).
    
    To bound the ratio $T_{sc}/T_{\text{aug}}$, we first consider the case of $w\gamma<1$ for which we have the inequalities
\begin{align}
    \frac{T_{sc}}{T_{\text{aug}}}&\le\frac{\frac{\binom{K}{K\gamma+1}-\binom{K-w}{K\gamma+1}}{\binom{K}{K\gamma}}}{\frac{1}{2}\frac{w(1-\gamma)}{(1+w\gamma)}}\le
    \frac{w(1-\gamma)}{\frac{1}{2}\frac{w(1-\gamma)}{(1+w\gamma)}}\le 4 ,\label{appBound1}
\end{align}
where we used the inequality $\frac{\binom{K}{K\gamma+1}-\binom{K-w}{K\gamma+1}}{\binom{K}{K\gamma}}\le w(1-\gamma)$ which we prove in Appendix~\ref{appBinomialsDifference}.

When $w\gamma\ge1$, the bound -- after a few basic algebraic manipulations -- takes the form
\begin{align}
     \frac{T_{sc}}{T_e}&=\frac{\frac{\binom{K}{K\gamma+1}-\binom{K-w}{K\gamma+1}}{\binom{K}{K\gamma}}}{\frac{1}{2}\frac{w(1-\gamma)}{(1+w\gamma)}}\le \frac{\frac{\binom{K}{K\gamma+1}}{\binom{K}{K\gamma}}}{\frac{1}{2}\frac{w(1-\gamma)}{(1+w\gamma)}}\\
    &=\frac{\frac{K(1-\gamma)}{1+K\gamma}}{\frac{1}{2}\frac{w(1-\gamma)}{1+w\gamma}}=2\frac{K(1+w\gamma)}{w(1+K\gamma)}=2+2\frac{K-w}{w\!+\!Kw\gamma}\\
    &<2\left(1+\frac{K+w}{w+wK\gamma}\right)<2\left(1+\frac{K}{wK\gamma}\right)\le 4, \label{app:3}
\end{align}
which concludes the proof.

\section{Conclusions and Ramifications}
In this work, we studied a cache-aided SISO BC in which users have different channel capacities. This model is motivated by the well-known worst-user bottleneck of coded caching, which, when left untreated, can severely deteriorate coded caching gains. The new algorithm establishes, together with the converse, the fundamental limits of performance within a factor of $4$, revealing that it is in fact possible to achieve the full-capacity performance even in the presence of many users with degraded link strengths. 

Pivotal to our approach is the identification of a `bottleneck (threshold) user', which may not necessarily be the user with the worst channel. From an operational point of view, this reveals that to increase performance, we must not necessarily focus on enhancing only the weakest users, but rather should focus on altering this bottleneck threshold. 

\appendices

\section{Bound on the difference of Binomials}\label{appBinomialsDifference}
Our aim is to prove the following corollary. 
\begin{corollary}
    For every integer $m>0$, the following inequality holds
    \begin{align}\label{eqAppendix2Inequality}
    \frac{\binom{K}{K\gamma+1}-\binom{K-m}{K\gamma+1}}{\binom{K}{K\gamma}}\le m(1-\gamma).
\end{align}
\end{corollary}
\begin{proof}
We first note that 
\begin{align}\label{eqSmallInequality}
    \frac{\binom{K-n-1}{K\gamma}}{\binom{K}{K\gamma}}\le (1-\gamma),~~n\ge 0
\end{align}
holds, because for any $m>p$, we have that $\binom{K-m}{K\gamma}< \binom{K-p}{K\gamma}$, which yields that Eq.~\eqref{eqSmallInequality} holds for any $n\ge 0$.

With this in place, in order to prove the inequality in Eq.~\eqref{eqAppendix2Inequality}, we employ proof by induction. Toward this, we first see that Eq.~\eqref{eqAppendix2Inequality} holds for $m=1$ because
\begin{align}
    \frac{\binom{K}{K\gamma+1}-\binom{K-1}{K\gamma+1}}{\binom{K}{K\gamma}}&=\frac{\binom{K}{K\gamma+1}-\frac{K-K\gamma-1}{K}\binom{K}{K-K\gamma-1}}{\binom{K}{K\gamma}}\\
    &=\frac{\binom{K}{K\gamma+1}}{\binom{K}{K\gamma}}\frac{K\gamma+1}{K}=(1-\gamma).
\end{align}
Now we assume that Eq.~\eqref{eqAppendix2Inequality} holds for some $n\ge 1$, and to prove that it also holds for $n+1$, we see that
\begin{align}\label{eqAppendixBStep1}
    \frac{\binom{K}{K\gamma+1}-\binom{K-n-1}{K\gamma+1}}{\binom{K}{K\gamma}}&=\frac{\binom{K}{K\gamma+1}-\binom{K-n}{K\gamma+1}+\binom{K-n-1}{K\gamma}}{\binom{K}{K\gamma}}\\ \label{eqAppendixBStep2}
    &=\frac{\binom{K}{K\gamma+1}-\binom{K-n}{K\gamma+1}}{\binom{K}{K\gamma}}\!+\!\frac{\binom{K-n-1}{K\gamma}}{\binom{K}{K\gamma}}\\ \label{eqAppendixBStep3}
    &\le n(1-\gamma)+(1-\gamma)
\end{align}
where in Eq.~\eqref{eqAppendixBStep1} we used the equality from Pascal's triangle, and then in Eq.~\eqref{eqAppendixBStep3} we used the inequality of Eq.~\eqref{eqAppendix2Inequality}. This concludes the proof.
\end{proof}

\enlargethispage{-1.2cm}

\end{document}